\RequirePackage{fix-cm}
\documentclass[smallcondensed]{svjour3}     \smartqed  
\usepackage{graphicx}
\usepackage{multicol}
\usepackage[english]{babel}
\usepackage{amsmath}
\usepackage{amsfonts}
\usepackage{amssymb}
\usepackage{subcaption}
\usepackage{array,xcolor,colortbl,booktabs}
\usepackage[Lenny]{fncychap}
\usepackage{newlfont}
\usepackage{lscape}
\usepackage{multirow}
\usepackage{lipsum}
\usepackage{url}
\usepackage{longtable}
\usepackage{booktabs,multirow}
\usepackage{tocloft}
\usepackage{soul}
\usepackage{ulem}
\usepackage{tabularx}
\usepackage{booktabs}
\newtheorem{Proposition}{Proposition}

\newtheorem{proof1}{Proof of Proposition}

\begin{document}

\title{A new class of composite indicators: the penalized power means}

\author{ Francesca Mariani $^{1*}$ \and Mariateresa Ciommi $^{1}$ \and Maria Cristina Recchioni $^{1}$       
}

\institute{ \at
             $^{1}$ \quad Department of Economic and Social Sciences; Università Politecnica delle Marche, Ancona 60121, Italy
              \email{m.c.recchioni@staff.univpm.it, m.ciommi@staff.univpm.it}       
                        \and
           $^*$Correspondence \at
               \email{f.mariani@staff.univpm.it} 
}

\date{Received: date / Accepted: date}

\maketitle

\begin{abstract}
In this paper we propose a new aggregation method for constructing  composite indicators that is based on a penalization of the power means. The idea underlying this approach consists in multiplying the power mean by a factor that takes into account for  the horizontal heterogeneity among indicators with the aim of penalizing the units with larger heterogeneity. In order to measure this heterogeneity, we scale the vector of normalized indicators by their power means, we compute the variance of the scaled normalized  indicators transformed by means of the appropriate Box-Cox function, and we measure the heterogeneity as the counter image of this variance through the Box-Cox function. The resulting penalization factor can be interpreted as the relative error, or the loss of information, that we obtain substituting the vector of the normalized indicators with their power mean. This penalization approach has the advantage to be fully data-driven and to be coherent with the same principle underlying the power mean approach, that is the minimum loss of information principle as well as to allow for a more refined rankings. The penalized power mean of order one coincides with the Mazziotta Pareto Index.
\end{abstract}
\section{Introduction}
The construction of composite indicators consists in reducing a multidimensional phenomenon into a one-dimensional phenomenon aggregating the multiple dimensions of the phenomenon, namely the indicators, into a single one, namely the composite indicator. The resulting composite indicator, although more simple and manageable to interpret and understand, is less informative with respect the vector of the indicators. That is, necessary, the aggregation process involves a loss of information. Despite the effort of many authors to develop objective measures of information loss (see, among others, \cite{Zhou2006}, \cite{Zhou2009}, \cite{Zhou2010}),  choosing effective measure of the information loss is a crucial task and depends on the subjective preference of the decision maker.  \\
\indent A good aggregation should conjugate the two-fold opposite objectives of reducing the dimension of the phenomenon under investigation as well as of generating a reasonable  loss of information. The majority of the aggregation functions are based on minimizing the loss for deviation of individual indicators from the aggregated value. In their paper  Calvo and Beliakov \cite{Calvo2010} systemize the various type of aggregation function, proving that any averaging aggregation functions are penalty-based functions, that is functions obtained minimizing the penalty  due to substitute the vectors of indicators with the aggregated value. Specifically, the power means are found choosing as penalty (loss) function the Euclidean distance from the vector of indicators transformed through a Box-Cox function (for more detail we refer to Berger and Casella \cite{Berger1992}. In other words, the power means can be seen as the least squares estimate of the vector of the indicators in the Box-Cox transformed space. Therefore, in the transformed space the power mean suffers from the same troubles of the arithmetic mean, that is the compensability and the substitutability. These issues, that are consequences of the loss of information, should be taken into account in the aggregation phase in order to differentiate units with same value of power mean.\\ 
\indent The idea of penalizing units in different way is shared by Mauro et al. \cite{Mauro2018} and Biggeri et al. \cite{Biggeri2019} that develop and apply in the well-being context the Multidimensional Synthesis of Indicators (MSI) approach, that aggregate the indicators relative to different units with a power mean with different order. The order of the power mean is a function of the arithmetic mean of the indicators, in a way such that units with lower indicator arithmetic mean have associated a lower order. Moreover, Rogge \cite{Rogge2018} applies a  Benefit of Doubt (BoD)-based weighted version of the power means for constructing an HDI index. The weights associated to the indicators depend on the unit in order to take into account for the difference in the indicator unbalance.    In line with the works of Mauro et al. \cite{Mauro2018}, Biggeri et al. \cite{Biggeri2019}  and Rogge \cite{Rogge2018} we propose a new aggregation approach that penalizes the power mean associated to a unit with a factor hat takes into account for  the horizontal heterogeneity among indicators with the aim of penalizing the units with larger heterogeneity. In order to measure this heterogeneity, we scale the vector of normalized indicators by their power means, we compute the variance of the scaled normalized  indicators transformed by means of the appropriate Box-Cox function, and we measure the heterogeneity as the counter image of this variance through the Box-Cox function. The resulting penalization factor can be interpreted as the relative error, or the loss of information, that we obtain substituting the vector of the normalized indicators with their power mean. The penalized power mean of order one coincides with the Mazziotta Pareto Index \cite{Mazziotta2016}. \\
\indent The penalization approach proposed in this paper is fully data-driven and can be effectively applied in many other fields, such as  environmental indices \cite{Sadiq2010}, fuzzy rule based systems, pattern recognition, decision making problems \cite{Khameneh2020} weighted voting systems \cite{Bustince2013}.\\
\indent The paper is organized as follows: in Section 1 we introduce the penalized power means, in Section 2 we prove some properties of this family. The Appendix contains the proof of the main proposition of Section 2.


\section{A new class of composite indicators}
Let $x_{ij}$ and $I_{ij}$ be, respectively, the value and  the normalized value of the indicator $j,$  $j=1,2,\ldots,m,$  relative to unit $i,$ $i=1,2,\ldots,m,$ such that $I_{ij}\in[0,1].$ Let $\underline{I}_i=[I_{i1}\,I_{i2}\,\cdots \, I_{im}]'$ be vector of indicators relative to the $i-$th unit, $i=1,2,\ldots,n,$ the power mean of order $p$ associated to $\underline{I}_{i}$  is  defined by:
\begin{eqnarray}\label{mean}
M_{p,i}=M_{p}(\underline{I}_i)=\left\{\begin{array}{ll}
\displaystyle \left(\frac{1}{m}\sum_{j=1}^m I_{ij}^p\right)^{\frac{1}{p}}, & p\neq 0,\\
\displaystyle  \left(\prod_{j=1}^m I_{ij}\right)^{\frac{1}{m}}, & p=0.
\end{array}\right.
\end{eqnarray}
The arithmetic mean, the geometric mean and the harmonic mean are special cases of the power mean for $p=1,$ $p=0$ and $p=-1,$ respectively. \\
\indent For $i=1,2,\ldots,n$ the composite indicator $M_{p,i}$  can be read as solution of the following optimization problem \cite{Berger1992}:
\begin{eqnarray}\label{opt}
\min_{c\in\mathcal A_p} F_p(c; \underline{I}_i),
\end{eqnarray}
where:
\begin{eqnarray}\label{lf}
F_p(c;\underline{I}_i)=\frac{1}{m}\sum_{j=1}^m (h_{p}(I_{ij})-h_{p}(c))^2
\end{eqnarray}
is the (information) loss function (or, the penalty function, to use the nomenclature of Calvo and Beliakov \cite{Calvo2010}) and  $h_p(x)$ is the Box-Cox transformation \cite{Box1964}:
\begin{eqnarray}\label{h}
h_{p}(x)=\left\{\begin{array}{ll}
\displaystyle \frac{x^p-1}{p}, & p\neq 0,\\
\ln x, & p=0.
\end{array}\right.
\end{eqnarray}
In (\ref{opt})  $\mathcal A_1=\mathbb R$ and $\mathcal A_p=\{x\in\mathbb R\ : \ x>0\}$   denote, respectively, the domain of the function $h_p$ for $p=1$ and $p\neq 1.$  Note that the function $h_p(1)=0$ and $h_p(x)$  is a strictly  increasing function of $x$. \\ 
\indent The function $F_p$ in (\ref{lf}) measures the sample (biased) variance of vector $h_p(\underline I_i)$ $= [h_p(I_{i1})\,h_p(I_{i2})\,\cdots \, h_p(I_{im})]'$ from its arithmetic mean:
\begin{eqnarray}
h_p(M_{p,i})=M_1(h_p(\underline{I}_i))=\frac{1}{m}\sum_{j=1}^m h_p(I_{ij}),
\end{eqnarray}
that is in the $p-$transformed space (the space obtained transformed the m dimensional vectors via the Box-Cox function of order $p$) the $p$-order generalized mean plays the role of arithmetic mean.\\
\indent Moreover, for any unit $i$ we can measure the error (loss of information) caused from substituting the vector of indicators $h_p(\underline{I}_i)$ with $h_p(M_{p,i})$ evaluating the objective function $F_p$ at its optimum:
\begin{eqnarray}
F_p(M_{p,i}; \underline{I}_i)=\frac{1}{m}\sum_{j=1}^m (h_{p}(I_{ij})-h_{p}(M_{p,i}))^2.
\end{eqnarray}
It is easy to see that this error coincides with the (biased) sample variance of the vector $h_p(\underline{I}_i),$ that here and in the rest of paper we denote by $S^2_{p,i}.$ Therefore  the quantity  $h_p^{-1}(S_{p,i}^2)$
 is a measure of the information loss caused from substituting  $\underline{I}_{i}$ with $M_{p,i}.$  Note that the size of $h_p^{-1}(S_{p,i}^2)$ depends strongly on $M_{p,i},$  
then the variances associated to  units with different power means are not comparable.  \\
\indent In order to  measure the relative information loss caused from substituting $\underline{I}_{i}$ with $M_{p,i},$ it is necessary to scale the indicators referring to a same unit by a specific
criterion that removes the dependence from the power mean. To this purpose,  let us  consider the vector of scaled normalized indicators, $\underline{\tilde I}_i=[\tilde{I}_{i1}\, \tilde{I}_{i2}, \cdots \, \tilde{I}_{in}]^\top,$ where:
\begin{eqnarray}\label{rel_err}
\tilde{I}_{ij}=\frac{I_{ij}}{M_{p,i}},\quad i=1,2,\ldots,n.
\end{eqnarray}
Note that from the homogeneity property of the $p$-order power means it follows that $M_p(\underline{\tilde I}_i)=1,$  for  $i=1,2,\ldots,n,$ then the the error (loss of information) caused from substituting the vector of indicators $h_p(\underline{\tilde I}_i)$ with $h_p(M_p(\underline{\tilde I}_i))=h_p(1)=0$ is given by:
\begin{eqnarray}\label{S}
\tilde S_{p,i}^2=\frac{1}{m}\sum_{j=1}^m (h_p(\tilde{I}_{ij})-h_p(1))^2=\frac{1}{m}\sum_{j=1}^m (h_p(\tilde{I}_{ij}))^2,\quad i=1,2,\ldots,n.
\end{eqnarray} 
For $i=1,2,\ldots,n,$ the quantity 
\begin{eqnarray}
 h_p^{-1}(\tilde S_{p,i}^2)=\left\{\begin{array}{ll}
\left(\displaystyle 1 + p\, \tilde S_{p,i}^2\right)^{\frac{1}{p}},& p\neq 0,\\
&\\
\exp\left(\tilde S_{0,i}^2\right), &p=0,
\end{array}\right.
\end{eqnarray}
 is independent from the size of $\tilde M_{p,i}$ and measures the relative information loss caused from substituting $\underline{I}_{i}$ with $M_{p,i}.$  The higher the value of $h_p^{-1}(\tilde S_{p,i}^2),$ the higher  the loss of information caused from considering $M_{p,i}$ instead of the sub-indicator vector $\underline{I}_{i},$ no matter the value of $M_{p,i}.$  \\ We use $h_p^{-1}(\pm \tilde S_{p,i}^2)$ to penalize the power mean of order $p.$ In particular, for $i=1,2,\ldots,n,$ the penalized power mean of order $p$ associated to the normalized indicator vector $\underline{I}_i$ is defined by: 
\begin{eqnarray}\label{pm}
PM_{p,i}^{\pm}=M_{p,i}h_p^{-1}(\pm \tilde S_{p,i}^2)=\left\{\begin{array}{ll}
M_{p,i}\left(\displaystyle 1 \pm p\, \tilde S_{p,i}^2\right)^{\frac{1}{p}},& p\neq 0,\\
&\\
M_{0,i}\exp\left(\pm \tilde S_{0,i}^2\right), &p=0.
\end{array}\right.
\end{eqnarray}
Note that the $\pm$ sign in (\ref{pm}) depends on the type of phenomenon considered, if increasing variations of the indicator correspond to positive variations of the phenomenon (positive polarity) we choose the sign $-,$ otherwise (negative polarity) we choose the sign $+.$ \\
\indent The scaling done in (\ref{rel_err}) ensures that the  term $h_p^{-1}(\pm\tilde S_{p}^2(\underline{I}_i))$ penalizes the score of each unit (the $p$-order power mean of the normalized indicators) independently from the value of the power mean itself with a quantity  that is directly proportional to the ``horizontal variability'' of the indicators. The aim of the penalization is to favour the units that, power mean being equal, have a greater balance among the indicators. This is the idea underlying the ``Method of Penalties by Coefficient of Variation”, introduced by Mazziotta and Pareto \cite{Mazziotta2016}, that  adjust the arithmetic mean by a penalization coefficient that is function, for each
unit, of the indicator coefficient of variation to define the Mazziotta Pareto Index (MPI).
\begin{Proposition}\label{prop1}
The penalized power mean of order one, $PM_{1,i}^{\pm},$  coincides with the MPI.  
\end{Proposition}
\begin{proof} For the proof  see the Appendix.
\end{proof}
\medskip
\begin{Proposition}\label{prop2}
The penalized power mean defined in (\ref{pm}) satisfies the following properties:
\begin{itemize}
\item[1.] $PM_{p,i}^{+}\geq M_{p,i} \geq PM_{p,i}^{-}$.\\
\item[2.] $PM_{p,i}^{+}=PM_{p,i}^{-}=M_{p,i}, \text{ if and only if }\tilde S_{p,i}=0$.\\
\item[3.] $(PM_{p,i}^+)^p=2(M_{p,i})^p-(PM_{p,i}^-)^p$ for $p\neq 0$. \\
\item[4.] $PM_{0,i}^+=PM_{0,i}^{-}\exp\left\{2\, \tilde S_{0,i}^2\right\}$.\\
\item[5.] Given two units $k$ and $h$ ($k\neq h$) with $M_{p,k}=M_{p,h},$ we have:
\begin{align*}
PM_{p,k}^->PM_{p,h}^-&\quad \text{ iff }\quad \tilde S_{p,h}^2>\tilde S_{p,k}^2,\\
PM_{p,k}^+>PM_{p,h}^+&\quad\text{ iff }\quad \tilde S_{p,k}^2>\tilde S_{p,h}^2.
\end{align*}
\item[6.] Given two units $k$ and $h$ ($k\neq h$) with $M_{p,k}>M_{p,h},$ for $p\neq 0,$ we have:
\begin{align*}
PM_{p,k}^->PM_{p,h}^-&\quad \text{ iff } \quad M_{p,k}^p-M_{p,h}^p>p\left(M_{p,k}^p \, \tilde S_{p,k}^2-M_{p,h}\, \tilde S_{p,h}^2\right),\\
PM_{p,k}^+>PM_{p,h}^+& \quad \text{ iff } \quad M_{p,k}^p-M_{p,h}^p>p\left(M_{p,h}^p \, \tilde S_{p,h}^2-M_{p,k}^p\, \tilde S_{p,k}^2\right).
\end{align*}
\item[7.] Given two units $k$ and $h$ ($k\neq h$) with $M_{0,k}>M_{0,h},$ we have:
\begin{align*}
PM_{0,k}^->PM_{0,h}^-&\quad \text{ iff } \quad \frac{M_{0,k}}{M_{0,h}}>\exp\left\{\tilde S_{0,k}^2-\tilde S_{0,h}^2\right\},\\
PM_{0,k}^+>PM_{0,h}^+&\quad \text{ iff } \quad \frac{M_{0,k}}{M_{0,h}}>\exp\left\{\tilde S_{0,h}^2-\tilde S_{0,k}^2\right\}.
\end{align*}
\item[8.] $\underset{p \to -\infty}{\lim}  PM_{p,i}^{\pm}=\underset{j=1,2,\ldots,m}{\min}I_{ij}.$
\item[9.] $\underset{p \to +\infty}{\lim}  PM_{p,i}^{\pm}=\underset{j=1,2,\ldots,m}{\max}I_{ij}.$
\end{itemize}
\end{Proposition}
\begin{proof} For the proof of Properties 8 and 9 see the Appendix.
\end{proof}
\medskip
Note that Properties 8 and 9 in Proposition \ref{prop2} imply that the penalization has not effect when the power mean of order $-\infty$ (i.e. the minimum function) or $+\infty$ (i.e. the maximum function) is considered.
In fact the minimum and the maximum functions realize already, respectively, in the case of positive polarity and in the case of negative polarity, the maximum penalization for unbalanced values of the indicators, therefore they do not need further penalizations. 
\begin{Proposition}\label{prop3}
The penalization factor in (\ref{pm}):
\begin{eqnarray}\label{p_factor}
g_{p,i}^{\pm}=(1\pm \tilde S_{p,i}^2)^{\frac{1}{p}}
\end{eqnarray}
satisfies the following properties:
\begin{itemize}
\item[1.] $\displaystyle\frac{ \partial g_{p,i}^+}{\partial p}<0$ for $p>0$ and $\displaystyle\frac{\partial g_{p,i}^+}{\partial p}>0$ for $p<0.$
\item[2.] $\displaystyle\frac{\partial g_{p,i}^-}{\partial p}>0$ for $p>0$ and $\displaystyle\frac{\partial g_{p,i}^+}{\partial p}<0$ for $p<0.$
\item[3.] $\underset{p \to 0}{\lim}g_{p,i}^{\pm}=\exp\left\{\pm\tilde S_{0,i}^2\right\}.$
\item[4.] $\underset{p \to \pm\infty}{\lim}g_{p,i}^{\pm}=1.$
\end{itemize}
\end{Proposition}
\begin{proof} See the Appendix. 
\end{proof}
\medskip
\begin{Proposition}\label{prop4}
For the penalized power mean (\ref{pm}) the Marginal Rate of Compensation (MRC) \cite{CasadioTarabusi2013} between variables $z_{i k},$ $z_{i h}:$
\begin{eqnarray}\label{MRC}
MRC_{kh,i}=\left.\frac{\partial PM_{p,i}}{\partial I_{ik}}\right/\frac{\partial PM_{p,i}}{\partial I_{ih}}=\left\{\begin{array}{ll}
\displaystyle \left(\frac{I_{ik}}{I_{ih}}\right)^{p-1}, &p\neq 0,\\
\\
\displaystyle \frac{I_{ih}}{I_{ik}}, & p=0.
\end{array}\right.
\end{eqnarray}
\end{Proposition}
\begin{proof} See the Appendix. 
\end{proof}

\newpage

\section{Appendix}

\begin{proof1}
Substituting (\ref{h}) for $p=1$ in (\ref{S}) we have:
\begin{align}\label{S1}
\tilde S_{1,i}^2=&\frac{1}{m}\sum_{j=1}^m (\tilde{I}_{ij}-1)^2=\frac{1}{m}\sum_{j=1}^m \left(\frac{I_{ij}}{M_{1,i}}-1\right)^2=\frac{\frac{1}{m}\sum_{j=1}^m (I_{ij}-M_{1,i})^2}{M_{1,i}^2}\nonumber\\
=&\frac{S_{1,i}^2}{M_{1,i}^2},\quad i=1,2,\ldots,n,
\end{align} 
where:
\begin{eqnarray}
S_{1,i}^2=\frac{1}{m}\sum_{j=1}^m (I_{ij}-M_{1,i})^2, \quad i=1,2,\ldots,n,
\end{eqnarray}
is the (biased) sample variance of vector $\underline{I}_i.$\\
Substituting (\ref{S1}) into (\ref{pm}) for $p=1$ we have:
\begin{eqnarray}\label{MPI}
PM_{1,i}^\pm=M_{1,i}\left(1\pm\frac{S_{1,i}^2}{M_{1,i}^2}\right),\quad i=1,2,\ldots,n,
\end{eqnarray}
that is the MPI.
\end{proof1}

\begin{proof1}[Properties 8 and 9]
The (biased) sample variance  $\tilde S_{p,i}^2$ in (\ref{pm}) can be rewritten as follows:
\begin{eqnarray}\label{cv1}
\tilde S_{p,i}^2=\frac{1}{p^2}\frac{1}{m}\sum_{j=1}^m \left(\left(\frac{I_{ij}}{M_{p,i}}\right)^p-1\right)^2
\end{eqnarray}
then, taking the limit of (\ref{cv1}) for $p\to +\infty$  and recalling that\\ $M_{p,i}\underset{p \to +\infty}{\longrightarrow}\max(I_{i1},I_{i2},\ldots,I_{im})\geq I_{ij},$ $j=1,2,\ldots,m,$ we obtain:
\begin{eqnarray}\label{lim1}
\underset{p \to +\infty}{\lim} p\, \tilde S_{p,i}^2= \underset{p \to +\infty}{\lim}\frac{1}{p}=0^+.
\end{eqnarray}
Substituting (\ref{lim1}) into (\ref{pm}) we obtain Property 7.\\
Analogously, recalling that $M_{p,i}\underset{p \to -\infty}{\longrightarrow}\min(I_{i1},I_{i2},\ldots,I_{im})\leq I_{ij},$ $j=1,2,\ldots,m,$ we have:
\begin{eqnarray}\label{lim2}
\underset{p \to -\infty}{\lim} p\,\tilde S_{p,i}^2= \underset{p \to -\infty}{\lim}\frac{1}{p}=0^-.
\end{eqnarray}
Substituting (\ref{lim2}) into (\ref{pm}) we prove Property 8.
\end{proof1}
\bigskip

\begin{proof1} 
The derivative of $g_{p,i}^{\pm}$ with respect to $p$ is:
\begin{eqnarray}\label{derg}
\frac{\partial g_{p,i}^{\pm}}{\partial p}=g_{p,i}\left[-\frac{1}{p^2}\ln\left(1\pm p\, \tilde S_{p,i}^2\right)\pm\frac{1}{p}\frac{1}{\left(1\pm p\, \tilde S_{p,i}^2\right)}\left(\tilde S_{p,i}^2+p\frac{\partial \tilde S_{p,i}^2}{\partial p}\right)\right].
\end{eqnarray}
The derivative of $\tilde S_{p,i}^2$ with respect to $p$ is:
\begin{eqnarray}\label{derS}
\frac{\partial \tilde S_{p,i}^2}{\partial p}=-\frac{2}{p}\tilde S_{p,i}^2+\frac{2}{p^2}\frac{1}{m}\underbrace{\sum_{j=1}^m\left(\left(\frac{I_{ij}}{M_{p,i}}\right)^p-1\right)}_{=0}\frac{\partial}{\partial p}\left(\left(\frac{I_{ij}}{M_{p,i}}\right)^p\right)=-\frac{2}{p}\tilde S_{p,i}^2.
\end{eqnarray}
Substituting (\ref{derS}) into (\ref{derg}), we obtain:
\begin{eqnarray}\label{derg1}
\frac{\partial g_{p,i}^{\pm}}{\partial p}=g_{p,i}\left[-\frac{1}{p^2}\ln\left(1\pm p\, \tilde S_{p,i}^2\right)\mp\frac{1}{p}\frac{\tilde S_{p,i}^2}{\left(1\pm p\, \tilde S_{p,i}^2\right)}\right].
\end{eqnarray}
Finally, observing that $\ln(1+ p\, \tilde S_{p,i}^2)>0$ and  $\ln(1- p\, \tilde S_{p,i}^2)<0,$ we obtain prove Properties 1 and 2.\\
Property 2 has been  proved in the proof of Proposition 1 (see Equations (\ref{lim3}), (\ref{lim4})).\\
Taking the limit for $p\to 0$ of $g_{p,i}^{\pm}$ we obtain:
\begin{align}\label{lim3}
\underset{p \to 0}{\lim} g_{p,i}^{\pm}=&\underset{p \to 0}{\lim}\left(1\pm p \tilde S_{p,i}^2\right)^{\frac{1}{p}}=\underset{p \to 0}{\lim}\exp\left\{\frac{\ln\left(1\pm p\tilde S_{p,i}^2\right)}{p}\right\}\nonumber\\
=&\underset{p \to 0}{\lim}\exp\left\{\tilde S_{p,i}^2\frac{\ln\left(1\pm p\right)}{p}\right\}
=\exp\left\{\underset{p \to 0}{\lim}\tilde S_{p,i}^2\,\frac{\ln\left(1\pm p\right)}{p}\right\}\nonumber\\=&
\exp\left\{\tilde S_{0,i}^2\right\}\,\exp\left\{\underset{p \to 0}{\lim}\frac{\ln\left(1\pm p\right)}{p}\right\}
\end{align}
and using the L'H\^opital's rule we obtain:
\begin{eqnarray}\label{lim4}
\underset{p \to 0}{\lim} g_{p,i}^{\pm}=\exp\left\{\pm \tilde S_{0,i}^2\right\}.
\end{eqnarray}
This concludes the proof.
\end{proof1}

\begin{proof1}
The derivative of $PM_{p,i}$ with respect to $I_{ik}$ is:
\begin{eqnarray}\label{derI}
\frac{\partial PM_{p,i}}{\partial I_{ik}}=\frac{\partial M_{p,i}}{\partial I_{ik}}\, g_{p,i}^{\pm}+M_{p,i}\, \frac{\partial g_{p,i}^{\pm}}{\partial I_{ik}}.
\end{eqnarray}
The derivative of $M_{p,i}$ with respect to $I_{ik}$ is:
\begin{eqnarray}\label{dermuz}
\frac{\partial M_{p,i}}{\partial I_{ik}}=\left\{\begin{array}{ll}
\displaystyle \frac{1}{m}p I_{ik}^{p-1},& p\neq 0,\\
\\
\displaystyle \frac{1}{m}\frac{1}{I_{ik}}, & p=0.
\end{array}\right.
\end{eqnarray}
The derivative of $g_{p,i}$ with respect to $I_{ik}$ is:
\begin{eqnarray}\label{dergz}
\frac{\partial g_{p,i}^{\pm}}{\partial I_{ik}}=\left\{\begin{array}{ll}
\displaystyle \frac{1}{p}\left(1\pm p\tilde S_{p,i}^2\right)^{\frac{1}{p}-1}\,\frac{\partial \tilde S_{p,i}^2}{\partial I_{ik}}, & p\neq 0,\\
0,& p=0.
\end{array}\right.
\end{eqnarray}
The derivative of $\tilde S_{p,i}^2$ with respect to $I_{ik}$ is:
\begin{eqnarray}\label{derSz}
\frac{\partial \tilde S_{p,i}^2}{\partial I_{ik}}=\frac{2}{p^2}\frac{1}{m}\underbrace{\sum_{j=1}^m\left(\left(\frac{I_{ij}}{M_{p,i}}\right)^p-1\right)}_{=0}\frac{\partial \left(\frac{I_{ik}}{M_{p,i}}\right)^p}{\partial I_{ik}}=0.
\end{eqnarray}
Substituting (\ref{dermuz}), (\ref{dergz}), (\ref{derSz})  into (\ref{derI}) we obtain:
\begin{eqnarray}\label{derI1}
\frac{\partial PM_{p,i}}{\partial I_{ik}}=\left\{\begin{array}{ll}
\displaystyle \frac{1}{m}p I_{ik}^{p-1}\, g_{p,i}^{\pm},& p\neq 0,\\
\\
\displaystyle \frac{1}{m}\frac{1}{I_{ik}}\, g_{p,i}^{\pm}, & p=0.
\end{array}\right.
\end{eqnarray}
From (\ref{derI1}) it follows easily (\ref{MRC}).
\end{proof1}

\newpage

\bibliographystyle{spmpsci}      
\bibliography{refrisks}   

\end{document}